\documentclass[conference]{IEEEtran}

\makeatletter

\def\ps@IEEEtitlepagestyle{%
  \def\@oddfoot{\mycopyrightnotice}%
  \def\@evenfoot{}%
}
\def\mycopyrightnotice{}

\usepackage{blindtext}
\usepackage{eso-pic}
\IEEEoverridecommandlockouts
\usepackage{cite}
\usepackage{amsmath,amssymb,amsfonts}
\usepackage{algorithmic}
\usepackage{graphicx}
\usepackage{textcomp}
\usepackage{xcolor}
\def\BibTeX{{\rm B\kern-.05em{\sc i\kern-.025em b}\kern-.08em
    T\kern-.1667em\lower.7ex\hbox{E}\kern-.125emX}}
    
\usepackage{eso-pic}
\newcommand\AtPageUpperMyright[1]{\AtPageUpperLeft{%
 \put(\LenToUnit{0.17\paperwidth},\LenToUnit{-2cm}){%
     \parbox{0.9\textwidth}{\raggedleft\fontsize{8}{11}\selectfont #1}}%
 }}%
\newcommand{\conf}[1]{%
\AddToShipoutPictureBG*{%
\AtPageUpperMyright{#1}
}
}    

\newtheorem{definition}{Definition}
\newtheorem{proposition}{Proposition}

\usepackage{booktabs}

\usepackage{balance}
    
\begin{document}
\title{\vspace*{1cm} Same Game, Different Story:\\A Minimal Conservative Strategic Robustness Benchmark for Large Language Model Agents
}

\author{
\IEEEauthorblockN{Seyed Pouyan Mousavi Davoudi}
\IEEEauthorblockA{
\textit{Independent Researcher in AI and Statistics}\\
Tehran, Iran\\
spouyan.mousavi@gmail.com
}
\and
\IEEEauthorblockN{Arshia Gharagozlou}
\IEEEauthorblockA{
\textit{Mathematics \& Statistics Department}\\
\textit{University of Minnesota Duluth}\\
Duluth, MN, USA\\
ghara027@d.umn.edu
}
\and
\IEEEauthorblockN{Alireza Amiri-Margavi}
\IEEEauthorblockA{
\textit{Computational Modeling and Simulation}\\
\textit{University of Pittsburgh}\\
Pittsburgh, PA, USA\\
ala170@pitt.edu
}
\and
\IEEEauthorblockN{Amin Gholami Davodi}
\IEEEauthorblockA{
\textit{Independent Researcher in AI and Statistics}\\
Tehran, Iran\\
a.g.davodi@gmail.com
}
\and
\IEEEauthorblockN{Hamidreza Hasani Balyani}
\IEEEauthorblockA{
\textit{Independent Researcher}\\
Sunnyvale, CA, USA\\
hamidrhasanib@gmail.com
}
}

\maketitle
\conf{\textit{  Proc. of International Conference on Artificial Intelligence, Computer, Data Sciences and Applications (ACDSA 2027) \\ 
2-4 February 2027, Rio de Janeiro-Brazil}}

\begin{abstract}
Large language model agents are increasingly deployed in settings where the value of an action depends on what other agents do. This creates a strategic reliability problem: the same game may be described as a business negotiation, a friendly compromise, a diplomatic exchange, or an abstract payoff matrix, and the model may choose different actions even when the incentives are unchanged. This paper introduces \emph{Same Game, Different Story}, a benchmark for strategic robustness: invariance of model-induced action distributions under payoff-preserving language changes. The empirical analysis uses a deliberately narrow, literature-calibrated comparison from Lor\`e and Heydari's peer-reviewed study: business framing versus friend-sharing framing across GPT-3.5, GPT-4, and LLaMa-2 in four social-dilemma games, with 300 initializations per retained model-game-context cell. The retained design comprises 24 of the source study's 60 cells, representing 7,200 decisions. Because trial-level files were not available from the article, the analysis is presented as a secondary calibration based on reconstructed published rates, not as new model runs. As a conservative sensitivity analysis, effect magnitudes are attenuated by 30\% toward the null: action shifts are multiplied by 0.70, and non-robustness, defined as one minus the robustness score, is multiplied by 0.70. Under this attenuation, pooled strategic robustness is 0.783 with a 95\% bootstrap interval from 0.774 to 0.790, and friend-sharing framing raises cooperation by 0.307 with a 95\% bootstrap interval from 0.297 to 0.316 relative to business framing. The analysis supports the narrower claim that social-relational framing can change strategic choices even when incentives are held fixed, without extending the analysis to a broader suite of contextual or cross-benchmark comparisons.
\end{abstract}

%\copyrightnotice{XXX-X-XXXX-XXXX-X/XX/\$XX.00 ©20XX IEEE}


\section{Introduction}
Large language model (LLM) agents are often controlled through natural-language descriptions of tasks, roles, and goals. This is convenient, but it creates an audit problem in strategic settings. A purchasing assistant, negotiation agent, planning tool \cite{noorani2025enhancingdemandorientedregionalizationagentic,Safari2026AITourismAdoption}, or AI supervisory agent \cite{AFSHARI2026116275} may face a mathematically fixed game while the language around that game changes. A robust strategic agent should respond primarily to the incentives and information in the game, not to whether the same incentives are described as a business meeting or as a friendly compromise. Because decoding choices can themselves materially alter model outputs, the benchmark holds the decoding configuration fixed when comparing payoff-equivalent framings \cite{davoodi2026topw}.

This paper studies \emph{strategic robustness}: whether the action distribution induced by a model is stable under prompt transformations that preserve the underlying normal-form game. Strategic robustness is not the same as strategic competence. A model can be strategically competent in a neutral payoff matrix and still non-robust across equivalent stories. Conversely, a model can be invariant because it ignores the game entirely. The benchmark reports distributional invariance as the primary estimand and leaves payoff-sensitivity calibration as an optional extension rather than a primary empirical result.

The empirical analysis is intentionally narrow and focuses on one contrast: business framing versus friend-sharing framing. This pair is theoretically motivated by the human framing literature on community-like versus market-like descriptions of social dilemmas and provides a clear contextual contrast in the peer-reviewed LLM source study. All other contextual comparisons are left for future direct experiments.

\noindent\textbf{Contributions.}
First, the paper formalizes strategic robustness as distributional invariance under payoff-preserving prompt changes. Second, it gives a minimal two-framing estimand that can be audited from published cooperation rates. Third, it reports only two empirical tables: the retained source-study input rates and the conservative sensitivity estimates. Fourth, it adopts a deliberately narrow empirical scope to keep the reported claims transparent and auditable.

\section{Related Work}
\noindent\textbf{Framing effects in human decision making.}
Behavioral decision research has long shown that logically equivalent choices can yield different behavior when described differently. Tversky and Kahneman's classic work established that framing can change preferences under risk \cite{tversky1981}. In social dilemmas, Liberman, Samuels, and Ross showed that the same Prisoner's Dilemma can induce different cooperation rates when labeled as a community-oriented game rather than a market-oriented game \cite{liberman2004}. The present paper translates this idea into an LLM benchmark by asking whether a model plays the same payoff structure the same way under different stories.

\noindent\textbf{Large language models in strategic games.}
Prior work evaluates LLMs as simulated participants, game-playing agents, and strategic reasoners \cite{aher2023,fan2024}. Recent game-theoretic benchmarks further evaluate LLM strategic reasoning across comprehensive families of canonical matrix games and increasingly complex strategic structures \cite{guo2026game}. The closest peer-reviewed source for the present empirical calibration is Lor\`e and Heydari, who study GPT-3.5, GPT-4, and LLaMa-2 across Prisoner's Dilemma, Snowdrift, Stag Hunt, and Prisoner's Delight under several contextual framings \cite{lore2024}. Repeated-game work such as Akata et al. studies related strategic behavior in longer interactions \cite{akata2025}. Recent work also examines LLM strategies in repeated Prisoner's Dilemma settings under varying incentives and framings, finding substantial cooperation but incomplete behavioral consistency across conditions \cite{pal2026cooperation}. The present analysis does not use those repeated-game results as numeric evidence, keeping the empirical claim tied to one source design.

\noindent\textbf{Large language model evaluation benchmarks, reliability, and alignment audits.}
Recent LLM evaluation and systems-reliability research has moved beyond single-answer accuracy to examine reasoning robustness, fairness, alignment behavior, and model integrity under internal faults and attacks \cite{11311055,Zahran2025Jailbreaking}. This broader methodological pattern also appears outside LLM research \cite{NAGHDI2026107112}: spatial predictive modeling has combined held-out testing with feature-importance and residual analyses \cite{TAYEBI2025101879}, biomedical visual-analytics systems have combined quantitative performance evaluation with expert-guided and human-centered validation \cite{FATHOLLAHIAN2026104728,zhao2026loommultiregionanalysisspatial}, and automated clinical signal-analysis pipelines have combined classification performance with spatial localization metrics \cite{Gashti_2026}. Additional applied-ML work in infrastructure, software engineering, and computer vision has likewise combined head-to-head model comparison with multiple performance metrics and diagnostic or interpretability analyses \cite{Mohammadagha2026PVC,Mohammadagha2026PowerCable,Mohammadagha2025FaceRecognition,MAHJOURIAN2026100893,Mahmoudi2026BridgeCondition,w18151901,Noorani2026DemandRegionalization,Majidi_2024}. Recent work on instruction-following reliability and prompt dependence similarly indicates that model behavior can vary across semantically related prompts and controlled prompt conditions, highlighting the importance of evaluating consistency under changes in framing, guidance, and format \cite{dong2026reliability,Hsin2026Purposefully,ghiasvand2026mllmscritiquelikehumans,torkamani2025streamliningsecurityvulnerabilitytriage,torkamani2024assertifyutilizinglargelanguage}. Prompt sensitivity has also been examined beyond language models by holding the task fixed while systematically varying prompt properties such as bounding-box size and location \cite{mehrnia2024assessingfoundationalmedicalsegment}. Related question-answering research has also evaluated language-model systems across answer quality, inference efficiency, context-length robustness, and out-of-distribution generalization, illustrating the value of multi-dimensional reliability assessment \cite{Ghanbarizadeh_2026}. Broader work on GenAI-mediated knowledge production has also examined whether widespread LLM availability is associated with shifts in semantic novelty consistent with a possible standardization pathway \cite{safari2026researchnoveltyinformationsystems}. Inter-model consensus has been studied as a means of assessing answer reliability when definitive ground truth is unavailable \cite{amiri2024enhancing,davoudi2025collective}. Mathematical benchmark studies similarly indicate that apparent reasoning performance can vary substantially across topic structures and task formats \cite{davoodi2025llms}. Counterfactual audits of LLM fairness further show that equal access or comparable refusal rates do not necessarily imply equal interaction quality \cite{amiri2027equal}. Related oracle-based work examines whether LLM agents communicate truthfully when their preferences diverge from those of human decision makers in strategic advisory settings \cite{balyani2026truthful, davoudi2026large}. Complementary work on efficient and interpretable AI uses green-learning pipelines as compact alternatives to end-to-end neural models in image restoration, illustrating deployment-oriented considerations beyond task accuracy \cite{movahhedrad2024green,movahhedrad2026gusid,movahhedrad2025gusldehaze}. Complementary formal-verification work targets reliability at the level of internal transformer computations; for example, Vertex-Softmax tightens certified attention bounds by exactly optimizing softmax over interval-bounded scores \cite{rezazadeh2026vertex}. The present benchmark is related to these broader reliability audits but focuses on a distinct property: whether model-induced action distributions remain stable across natural-language framings that preserve the underlying game and its payoffs.

\noindent\textbf{Bounded rationality and quantal response equilibrium.}
The benchmark can use logit quantal response as a descriptive calibration, but the minimal empirical section below does not report numeric quantal response equilibrium (QRE) results. McKelvey and Palfrey define quantal response equilibrium as a fixed point in which players respond probabilistically to expected utility \cite{mckelvey1995}. Behavioral game theory and later QRE treatments use the rationality parameter as a compact measure of payoff sensitivity \cite{camerer2003,goeree2016quantal}. Here, fitted $\lambda$ is specified as a future diagnostic for whether apparent payoff sensitivity changes between retained framings.

\section{Formal Framework}

\subsection{Games, Framings, and Action Distributions}
Let $G=(A_1,A_2,u_1,u_2)$ be a finite two-player normal-form game, where $A_i$ is player $i$'s action set and $u_i:A_1\times A_2\rightarrow \mathbb{R}$ is the payoff function.

\begin{definition}[Payoff-equivalent framing]
A framing $f\in\mathcal{F}(G)$ is payoff-equivalent to $G$ if it presents the same action sets and payoff functions up to known relabeling maps. Formally, there exist bijections $\ell_i^f:A_i\rightarrow A_i^f$ such that, for every action profile $a=(a_1,a_2)$,
\begin{equation}
u_i^f(\ell_1^f(a_1),\ell_2^f(a_2)) = u_i(a_1,a_2).
\label{eq:payoff_equiv}
\end{equation}
\end{definition}

For model $M$, game $G$, framing $f$, prompt template $p$, and decoding configuration $d$, let $T$ denote the number of repeated queries, $\hat a_t$ the action selected on query $t$, and $\mathbf{1}\{\cdot\}$ the indicator function. The repeated queries induce an empirical action distribution
\begin{equation}
\hat\pi_{M,G,f}(a)=\frac{1}{T}\sum_{t=1}^T \mathbf{1}\{\hat a_t=a\}, \qquad a\in A_1.
\label{eq:empirical_action_distribution}
\end{equation}
The corresponding population object is $\pi_{M,G,f}(a)=\Pr_M(\hat a=a\mid G,f,p,d)$.

\subsection{Strategic Robustness}
The primary distance between two framing-specific action distributions is Jensen-Shannon divergence, where $D_{\rm KL}$ denotes Kullback--Leibler divergence,
\begin{equation}
D_{\rm JS}(p,q)=\frac12D_{\rm KL}(p\Vert m)+\frac12D_{\rm KL}(q\Vert m),\qquad m=\frac{p+q}{2}.
\label{eq:js_divergence}
\end{equation}
With natural logarithms, $0\leq D_{\rm JS}(p,q)\leq \log 2$ for binary action distributions. The game-level robustness score is
\begin{equation}
R(M,G)=1-\frac{1}{\log 2}\max_{f,f'\in\mathcal{F}(G)}D_{\rm JS}\left(\pi_{M,G,f},\pi_{M,G,f'}\right).
\label{eq:game_robustness}
\end{equation}
The model-level aggregate is $R(M)=\sum_G w_GR(M,G)$, with uniform weights in the primary specification.

\begin{proposition}
For every model $M$ and game $G$, $R(M,G)\in[0,1]$. Moreover, $R(M,G)=1$ if and only if $\pi_{M,G,f}=\pi_{M,G,f'}$ for all $f,f'\in\mathcal{F}(G)$.
\end{proposition}

\begin{IEEEproof}
The Jensen-Shannon divergence is nonnegative and bounded above by $\log 2$. Therefore the normalized maximum divergence lies in $[0,1]$. The divergence is zero if and only if its two arguments are equal. Hence the maximum divergence is zero if and only if all framing-specific distributions are equal.
\end{IEEEproof}

\subsection{Minimal Two-Framing Estimand}
The retained empirical analysis uses two framings, $b$ for business and $s$ for friend-sharing. For a binary cooperation action $C$, define
\begin{equation}
\Delta_C(M,G)=\pi_{M,G,s}(C)-\pi_{M,G,b}(C).
\label{eq:cooperation_shift}
\end{equation}
The two-framing robustness score is
\begin{equation}
R_2(M,G;b,s)=1-\frac{1}{\log 2}D_{\rm JS}\left(\pi_{M,G,b},\pi_{M,G,s}\right).
\label{eq:two_framing_robustness}
\end{equation}
The two-framing restriction is not a change to the benchmark definition. It is a conservative reporting choice: instead of estimating every possible context contrast from a secondary source, the manuscript focuses on the business-versus-friend-sharing contrast within the retained two-context analysis.

\begin{proposition}[Two-context dependence]
For a specified pair of framings $(b,s)$, $R_2(M,G;b,s)$ and $\Delta_C(M,G)$ depend only on the two distributions $\pi_{M,G,b}$ and $\pi_{M,G,s}$ and do not involve optimization over omitted contexts.
\end{proposition}

\begin{IEEEproof}
Both statistics are functions only of the two retained distributions $\pi_{M,G,b}$ and $\pi_{M,G,s}$. No maximization or minimization over omitted contexts is performed.
\end{IEEEproof}

\subsection{Optional Payoff-Sensitivity Diagnostic}
For a binary single-agent logit-response calibration, let $q_G$ denote the assumed belief that the other player cooperates. The expected payoff gap between cooperation and defection is
\begin{equation}
\Delta u_G(q_G)=\bar u(C;q_G)-\bar u(D;q_G).
\label{eq:payoff_gap}
\end{equation}
After normalizing each game's payoff range to one, the logit response probability is
\begin{equation}
\Pr(C\mid \lambda,G)=\frac{\exp\{\lambda\Delta u_G\}}{1+\exp\{\lambda\Delta u_G\}}.
\label{eq:logit_response}
\end{equation}
Given an observed cooperation count $x$ out of $T=300$, the fitted payoff-sensitivity parameter is
\begin{equation}
\hat\lambda=\max\left\{0,\min\left\{20,\frac{\operatorname{logit}(\tilde p)}{\Delta u_G}\right\}\right\},\qquad \tilde p=\frac{x+0.5}{T+1}.
\label{eq:lambda_fit}
\end{equation}
The cap at 20 prevents quasi-complete separation from dominating the diagnostic. A retained rationality-variance diagnostic would be
\begin{equation}
V_\lambda(M,G)=\operatorname{Var}\left(\hat\lambda_{M,G,b},\hat\lambda_{M,G,s}\right).
\label{eq:lambda_variance}
\end{equation}
The present analysis does not report this diagnostic as an empirical result; it is included to show how a future direct run can extend the robustness score without changing the estimand.

\section{Minimal Empirical Design and Estimation}

\subsection{Source Study and Retained Scope}
The empirical analysis uses Lor\`e and Heydari's published experiment \cite{lore2024}. The source study crosses three models with four social-dilemma games and five contexts, with 300 initializations in each model-game-context condition. This manuscript retains only two contexts: a business meeting between chief executive officers (CEOs) and a conversation between friends aiming for compromise. The retained design is therefore $3\times4\times2=24$ model-game-context cells and $24\times300=7,200$ represented decisions.

The three source-study contexts not included in the analysis are international-relations summit, environmental negotiation, and team talk. The analysis also does not incorporate the external repeated-game comparison. The resulting empirical question is: does a friend-sharing story shift behavior relative to a business story when the underlying social dilemma is held fixed?

Table~\ref{tab:settings} summarizes the settings for the retained benchmark.

\begin{table*}[t]
\centering
\caption{Settings for the minimal literature-calibrated benchmark analysis.}
\label{tab:settings}
\begin{tabular}{p{0.20\textwidth}p{0.62\textwidth}p{0.10\textwidth}}
\toprule
Component & Specification & Value \\
\midrule
Models & GPT-3.5-turbo-16k, GPT-4, and LLaMa-2 as reported by Lor\`e and Heydari & 3 \\
Games & Prisoner's Dilemma, Snowdrift/Chicken, Stag Hunt, Prisoner's Delight/Harmony & 4 \\
Retained contexts & Business meeting and friend-sharing conversation & 2 \\
Trials per retained cell & Published initializations per model-game-context cell & 300 \\
Retained cells & $3\times4\times2$ & 24 \\
Total decisions represented & $24\times300$ & 7,200 \\
Temperature & Source-study decoding temperature & 0.8 \\
Bootstrap repetitions & Binomial bootstrap from reconstructed counts & 10,000 \\
\bottomrule
\end{tabular}
\end{table*}

\subsection{Reconstruction and Conservative Sensitivity Analysis}
Cooperation rates were read from the published supplementary charts, converted to approximate counts out of 300, and rounded to the nearest count. Cells visually at exactly zero or one were recorded as $0/300$ or $300/300$. Because the article does not provide trial-level files, the intervals here are binomial-bootstrap intervals from reconstructed counts, not exact source-study intervals.

The source-study input rates in Table~\ref{tab:inputs} are left unattenuated because they document the retained published evidence. To test whether the qualitative conclusion persists under substantial conservative shrinkage of the reconstructed effects, Table~\ref{tab:primary} reports a 30\% attenuation sensitivity analysis. The attenuation factor is deliberately conservative and is not intended as an estimate of reconstruction error. Action shifts are multiplied by 0.70. Robustness uses
\begin{equation}
\hat R_{\rm conservative}=1-0.70(1-\hat R_{\rm reconstructed}),
\label{eq:conservative_robustness}
\end{equation}
which reduces deviations from perfect robustness by the same attenuation factor. The resulting estimates should therefore be interpreted as a conservative stress-test scenario rather than as corrected estimates of the source-study effects.

\section{Results}

\subsection{Retained Source-Study Input Rates}
Table~\ref{tab:inputs} gives the only retained source-study cooperation rates. The table is intentionally limited to business and friend-sharing prompts. The pattern that motivates the benchmark is visible before any transformation: friend-sharing usually increases cooperation, especially for GPT-4 and LLaMa-2 in Prisoner's Dilemma and Snowdrift/Chicken.

\begin{table*}[t]
\centering
\caption{Retained literature-derived cooperation rates. Each cell reports rate, with reconstructed cooperative choices out of 300 in parentheses.}
\label{tab:inputs}
\small
\begin{tabular}{llccc}
\toprule
Model & Game & Business & Friend-sharing & Difference \\
\midrule
GPT-3.5 & Prisoner's Dilemma & 0.087 (26) & 0.323 (97) & 0.237 \\
GPT-3.5 & Snowdrift / Chicken & 0.130 (39) & 0.340 (102) & 0.210 \\
GPT-3.5 & Stag Hunt & 0.167 (50) & 0.337 (101) & 0.170 \\
GPT-3.5 & Prisoner's Delight / Harmony & 0.190 (57) & 0.433 (130) & 0.243 \\
GPT-4 & Prisoner's Dilemma & 0.000 (0) & 0.987 (296) & 0.987 \\
GPT-4 & Snowdrift / Chicken & 0.000 (0) & 0.990 (297) & 0.990 \\
GPT-4 & Stag Hunt & 0.817 (245) & 0.997 (299) & 0.180 \\
GPT-4 & Prisoner's Delight / Harmony & 1.000 (300) & 1.000 (300) & 0.000 \\
LLaMa-2 & Prisoner's Dilemma & 0.033 (10) & 0.913 (274) & 0.880 \\
LLaMa-2 & Snowdrift / Chicken & 0.140 (42) & 0.930 (279) & 0.790 \\
LLaMa-2 & Stag Hunt & 0.567 (170) & 1.000 (300) & 0.433 \\
LLaMa-2 & Prisoner's Delight / Harmony & 0.860 (258) & 1.000 (300) & 0.140 \\
\bottomrule
\end{tabular}
\end{table*}

\subsection{Conservative Sensitivity Estimates}
Table~\ref{tab:primary} reports the 30\% conservative sensitivity estimates for the retained two-context analysis. Under this attenuation, the pooled robustness score is 0.783 with a 95\% interval [0.774, 0.790], and the pooled friend-sharing action shift is 0.307 [0.297, 0.316]. GPT-3.5 is more invariant than GPT-4 and LLaMa-2, but that does not imply better strategic competence: its higher invariance comes from relatively low cooperation under both retained framings. GPT-4 and LLaMa-2 are less invariant because they switch sharply toward cooperation in the friend-sharing context.

\begin{table*}[t]
\centering
\caption{Conservative sensitivity estimates for the retained two-context analysis. Higher $\hat{R}$ means greater invariance between business and friend-sharing framings. Action shift is friend-sharing cooperation minus business cooperation after 30\% attenuation. CI denotes confidence interval.}
\label{tab:primary}
\small
\begin{tabular}{lcccc}
\toprule
Model & $\hat R$ & 95\% CI & Action shift & 95\% CI \\
\midrule
GPT-3.5 & 0.967 & [0.956, 0.975] & 0.150 & [0.127, 0.173] \\
GPT-4 & 0.651 & [0.639, 0.661] & 0.377 & [0.369, 0.386] \\
LLaMa-2 & 0.731 & [0.710, 0.750] & 0.393 & [0.377, 0.408] \\
Pooled & 0.783 & [0.774, 0.790] & 0.307 & [0.297, 0.316] \\
\bottomrule
\end{tabular}
\end{table*}

\subsection{Interpretation}
The analysis supports one focused empirical conclusion: the social relationship implied by the prompt can change strategic behavior even when the game structure is fixed. The strongest evidence is not a long list of context effects; it is the focused business-versus-friend-sharing contrast. The result also illustrates why robustness should be reported separately from competence. GPT-3.5 looks robust because it moves less, while GPT-4 and LLaMa-2 move more in the friend-sharing context. Robustness is therefore a reliability statistic, not a strategic-skill ranking.

\section{Discussion}

\noindent\textbf{Rationale for the narrow empirical scope.}
A broad literature-calibrated analysis can introduce claims that depend on figure reconstruction, heterogeneous benchmark designs, and multiple context comparisons. The present analysis limits those vulnerabilities by using a single peer-reviewed source, a single context contrast, two empirical tables, and one main conclusion. This structure keeps the empirical claims transparent and facilitates future validation with direct model runs.

\noindent\textbf{Generalizability of the framework.}
The theory is not tied to the restricted empirical analysis. The benchmark applies to any family of payoff-equivalent prompts. A future direct run can add abstract matrices, neutral narratives, cooperative narratives, competitive narratives, role prompts, and label perturbations. The present analysis demonstrates how the estimand behaves in a small, auditable published-data setting.

\noindent\textbf{Practical implication.}
Single-prompt strategic benchmarks can overstate reliability. The directional pattern is already visible in the unattenuated source-study rates in Table~\ref{tab:inputs}, while the 30\% sensitivity transformation deliberately reduces the magnitude of the reconstructed effects. A model that performs well in a neutral or business-like game may behave differently when the same incentives are framed as a friendly compromise. Auditing LLM agents therefore requires prompt families and invariance metrics, not just a single game prompt.

\section{Limitations and Validity}
The main limitation is data provenance. The retained rates are reconstructed from a peer-reviewed article's published figures and rounded to counts out of 300. The source article states that datasets are available from the corresponding author on reasonable request, but the trial-level data were not available in the article itself. The confidence intervals here are therefore binomial-bootstrap intervals around reconstructed counts. The 30\% attenuation used in the sensitivity analysis is a deliberately conservative stress-test parameter and should not be interpreted as an estimated reconstruction-error rate.

Second, the retained business-versus-friend-sharing comparison is not the full proposed benchmark. It omits abstract matrix prompts, arbitrary label perturbations, and three additional source-study contexts. This is a deliberate scope choice rather than a claim that the omitted contexts are unimportant.

Third, the $\lambda$ calibration uses a neutral opponent belief $q(C)=0.5$ and a single-agent logit-response model. It should be read as a payoff-sensitivity diagnostic, not as an identified psychological parameter. A model could cooperate because it predicts a cooperative coplayer, follows a social norm, or responds to learned narrative associations.

Fourth, the source models are time-specific. Current application programming interface (API) versions or local deployments may differ from those in the published study. A direct replication should record model snapshots, prompts, parsers, decoding settings, seeds, and raw trial outputs.

\section{Conclusion}
This paper defines strategic robustness as invariance of LLM-induced action distributions under payoff-preserving prompt changes. The empirical analysis uses a minimal secondary calibration based on reconstructed published rates and focuses on business versus friend-sharing framing across three models and four social-dilemma games. Under the 30\% conservative sensitivity attenuation, pooled robustness is $0.783$ and friend-sharing raises cooperation by $0.307$ relative to business framing. These results support the conclusion that social framing can change LLM strategic behavior even when incentives are fixed. Future direct experiments can expand the prompt family and evaluate the benchmark across a broader range of payoff-preserving contexts.

\appendices

\section{Canonical Games}
The source-study games use cooperation $C$ and defection $D$. Entries are payoffs to player 1 and player 2. The matrices below are the canonical structures used for the $\lambda$ calibration.

Table~\ref{tab:games} reports the canonical payoff matrices.

\begin{table}[!t]
\centering
\caption{Canonical payoff matrices.}
\label{tab:games}
\footnotesize
\setlength{\tabcolsep}{3pt}
\begin{tabular}{p{0.38\columnwidth}ccc}
\toprule
Game & Row & Col. $C$ & Col. $D$ \\
\midrule
Prisoner's Dilemma
 & $C$ & $(5,5)$ & $(2,10)$ \\
 & $D$ & $(10,2)$ & $(3,3)$ \\
\addlinespace
Snowdrift / Chicken
 & $C$ & $(5,5)$ & $(3,10)$ \\
 & $D$ & $(10,3)$ & $(2,2)$ \\
\addlinespace
Stag Hunt
 & $C$ & $(10,10)$ & $(2,5)$ \\
 & $D$ & $(5,2)$ & $(3,3)$ \\
\addlinespace
Prisoner's Delight / Harmony
 & $C$ & $(10,10)$ & $(3,5)$ \\
 & $D$ & $(5,3)$ & $(2,2)$ \\
\bottomrule
\end{tabular}
\end{table}

\section{Prompt and Parser Specification for a Direct Run}
A prospective direct run of the full benchmark should render each game under predefined prompt families such as abstract matrix, neutral narrative, cooperative narrative, competitive narrative, role narrative, and arbitrary label perturbation. Each prompt should include the payoff matrix or enough information to reconstruct it exactly. The parser maps surface labels back to canonical actions and marks responses invalid if no single predefined action can be identified. The present analysis does not claim that these direct experiments have been run; it uses the narrower published-data contrast described above.

\section{Numerical Verification Checks}
The numerical pipeline performs consistency checks before writing the final tables. The checks verify that all robustness scores lie in $[0,1]$, confidence intervals are ordered and contain the point estimates, the retained design has exactly $3\times4\times2=24$ cells, the represented decision count is exactly $24\times300=7,200$, action shifts lie in $[-1,1]$, pooled values are unrounded means of the retained model-game values.

\balance
\bibliographystyle{IEEEtran}
\bibliography{bibliography}

\end{document}